\documentclass[english,11pt]{article}
 \pdfoutput=1

\usepackage{amsthm,amsmath, amssymb}
\usepackage[authoryear]{natbib}
\usepackage[colorlinks=true,citecolor=blue, urlcolor=blue,breaklinks]{hyperref}
\usepackage{babel}
\usepackage{caption}
\usepackage{subcaption}
\usepackage{graphicx} 
\usepackage{color}
\usepackage{framed}
\usepackage{lscape}
\usepackage{rotating}
\usepackage{algorithm}
\usepackage[noend]{algpseudocode}

\usepackage{authblk}
\newtheorem{theorem}{Theorem}[section]

\newtheorem{lemma}[theorem]{Lemma}

\newtheorem{proposition}[theorem]{Proposition}

\usepackage{geometry} 
\makeatletter
\def\BState{\State\hskip-\ALG@thistlm}
\makeatother

\begin{document}

\title{Optimal Multiple Testing Under a Gaussian Prior on the Effect Sizes}
\author[1]{Edgar Dobriban\thanks{corresponding author: dobriban@stanford.edu}}
\author[2]{Kristen Fortney}
\author[2]{Stuart K. Kim}
\author[1]{Art B. Owen\thanks{owen@stanford.edu}}
\affil[1]{Department of Statistics, Stanford University, Stanford, California 94305, U.S.A.}
\affil[2]{Department of Developmental Biology, and Department of Genetics, Stanford University, Stanford, California 94305, U.S.A. }

\renewcommand\Authands{ and }

\maketitle

\begin{abstract}
We develop a new method for frequentist multiple testing with Bayesian prior information. Our procedure finds a new set of optimal p-value weights called the Bayes weights. Prior information is relevant to many multiple testing problems. Existing methods assume fixed, known effect sizes available from previous studies. However, the case of uncertain information is usually the norm. For a Gaussian prior on effect sizes, we show that finding the optimal weights is a non-convex problem.
Despite the non-convexity, we give an efficient algorithm that solves this problem nearly exactly. We show that our method can discover new loci in genome-wide association studies. On several data sets it compares favorably to other methods. Open source code is available.

\end{abstract}

\tableofcontents

\bigskip

\section{Introduction}

We are motivated by the genetics of human longevity.  Genome-wide association studies of longevity compare long-lived individuals (e.g. centenarians who live to 100 or older) to matched controls \citep{brooks2013genetics}. More than 500,000 genetic variants are tested for association to longevity. This is a large multiple testing problem. In addition to the multiplicity, the sample size is low - usually a few hundred. As a consequence, only a few loci have been replicably associated to human longevity. They do not explain the heritability of the trait \citep{iachine2006genetic}.

The multiplicity may be countered by testing only a few candidate variants selected based on prior scientific knowledge. In a separate work in preparation, led by Dr. Kristen Fortney, we find that a more general genome-wide test helps improve power in longevity. We leverage prior information from genome-wide association studies of age-related diseases, such as coronary artery disease and diabetes. For this task, we develop a new method of frequentist multiple testing with Bayesian prior information. In this paper we provide the theory for this method.

Our method is a type of p-value weighting scheme. P-value weighting is a general methodology for multiple testing that leverages independent prior information to improve power (see \cite{roeder2009genome,gui2012weighted} for a review). Suppose we test hypotheses $H_i$ $i = 1, 2, \ldots J$, via the p-values $P_i$. For a significance level $q \in [0,1]$, the weighted Bonferroni method declares the $i$-th hypothesis significant if $P_i \le q w_i$. The weights $w_i \ge 0$ are based on independent data. For weights averaging to 1, the family-wise error rate (the probability of making at least one error) is controlled strongly at $\alpha = Jq$. 

Previous work has found the optimal weights in a Gaussian model of hypothesis testing. Let the test statistics in the current study be $T_i \sim \mathcal{N}(\mu_i, 1)$, where $\mu_i$ are the \emph{means}, or \emph{effect sizes}, and test the null hypotheses $\mu_i \ge 0$ against $\mu_i <0$. We have some information about $\mu_i$ from prior studies. The works \cite{wasserman2006weighted, roeder2009genome, rubin2006method} consider the model where $\mu_i$ is known \emph{exactly} from the prior data, and the weights are allowed to depend on $\mu_i$. In this model they find the optimal weights maximizing the expected number of discoveries. We show this amounts to solving a convex optimization problem, but this was not used originally.

The assumption that $\mu_i$ are known precisely is problematic: if they were known, there should be no follow-up study. Instead, we account for uncertainty by considering the model $\mu_i \sim \mathcal{N}(\eta_i, \sigma_i^2)$. Here only the prior mean $\eta_i$ and standard error $\sigma_i$ are known from independent data, not the precise value of the effect sizes. However, finding the optimal weights, which we call Bayes weights, becomes a non-convex optimization problem. \cite{westfall1998using} use a direct numerical solver, and optimize for at most $4$ tests. 

We give an efficient method to find the optimal weights for a large number of tests. We solve the optimization problem exactly for small $q$. For larger $q$, we can solve it for a nearby $q^*$ such that $|q^*-q|\le 1/(2J)$. The cost in the first case is $O(J)$, the cost in the second case is $O(J\log J)$.  These are the costs per iteration of our optimization algorithm. We observe a near constant number of iterations used. The solution of non-convex optimization problems is challenging in general, thus it is perhaps remarkable that this problem admits a nearly exact solution.

This enables a new methodology for large-scale multiple testing that controls a frequentist error measure, while also taking into account Bayesian prior information. This method follows George Box's advice to be Bayesian when predicting but frequentist when testing \citep{box1980sampling}.  While this methodology was considered previously for a few tests \citep{westfall1998using}, we are the first to do it on a large scale. 

When prior information is uncertain, we show in simulations that the new scheme does better than its competitors. We also show theoretically, in a sparse mixture model, that weighting leads to substantially improved power. We apply our method to genome-wide association studies (GWAS). By analyzing several GWAS data sets we show its advantages compared to other methods.

This method should be useful for other problems in biology and elsewhere. Open source code is available from the authors. All our computational results are reproducible (see the Supplement). 

The contents of the paper are as follows: We discuss related work in Section \ref{sec:rel}. We develop the theory in Section \ref{sec:thy} and present simulations comparing our method to alternatives in Section \ref{sec:sim}. We apply our method to genome-wide association studies in Section \ref{sec:igwas} and use it to analyze GWAS data in Section \ref{sec:data_analysis}. The supplementary material contains a description of available software and instructions to reproduce our computational results (Section \ref{software}), as well as mathematical proofs (Section \ref{proofs}). 

\section{Related Work}
\label{sec:rel}

There is a large literature on related statistical methods for multiple testing with prior information. In early work, \cite{spjotvoll1972optimality} devised optimal multiple testing procedures treating tests unequally. Later it was recognized that Spjotvoll's results are equivalent to optimal p-value weighting methods. For instance, \cite{benjamini1997multiple} developed extensions of Spjotvoll's methods for p-value weighting. 

Leveraging Spjotvoll's results, \cite{wasserman2006weighted, roeder2009genome, rubin2006method} found the explicit formula for optimal weights in the Gaussian model $\mathcal{N}(\mu_i,1)$, assuming the effects are known exactly. This lead to an efficient method suitable for large applications. In the bioinformatics community \cite{eskin2008increasing, darnell2012incorporating} applied \cite{wasserman2006weighted}'s framework to GWAS. They accounted for correlations between the tests but assumed the effects are known exactly. 

The simple method of testing the top candidates from a prior study is also popular.  This is often known as two-stage testing or as a candidate study. A specific version for GWAS has been called the proxy-phenotype method \citep{rietveld2014common}. Using only the top candidates runs the risk of discarding many potentially useful hypotheses.


A missing ingredient is taking uncertainty into account. \cite{westfall1998using} considered a Gaussian model $\mathcal{N}(\mu_i,1)$ for hypothesis testing where prior distributions are known for the means. However, their optimization methods could handle only a small number ($J=4$) of tests.

\section{Theoretical results}
\label{sec:thy}

\subsection{Background}

In this section we present our theoretical results. As background, we begin with the case of known means. We work in the Gaussian means model of hypothesis testing:  We observe test statistics $T_i \sim \mathcal{N}(\mu_i, 1)$ and test each null hypothesis $H_i: \mu_i \ge 0$ against $\mu_i <0$.  The p-value for testing $H_i$ is $P_i = \Phi(T_i)$, where $\Phi$ is the normal cumulative distribution function. 

For a weight vector $w \in [0,\infty)^{J}$ and significance level $q \in [0,1]$, the weighted Bonferroni procedure rejects $H_{i}$ if $P_i \le q w_i$. Usual Bonferroni corresponds to $w_i=1$. Then the expected number of false rejections, known as the per-family error rate, equals: $\sum_{\mu_i \ge 0} \mathrm{pr}(P_i \le q w_i)   = q\sum_{\mu_i \ge 0} w_i.$ Therefore, if $\sum_{i=1}^{J} w_i  \le J$, the expected number of false rejections is controlled strongly, under any configuration of truth or falsehood of $H_i$, at $\alpha: = Jq$. By Markov's inequality this implies that the family-wise error rate, the probability of any false rejection, is also controlled at $\alpha$. This does not need independence of the $T_i$. We always assume that $q \le 1$, and usually $q \ll 1$. Without loss of generality we restrict the weights to $[0,1/q]$.

Let us denote the number of rejections by $R(w) = \sum_{i=1}^{J} I(P_i \le q w_i) $, where $I(\cdot)$ is the indicator function. The optimal weights in this model were found explicitly by \cite{wasserman2006weighted, roeder2009genome} and independently by \cite{rubin2006method}. Denoting by $E_T(\cdot)$ expectation with respect to $T_i$, they solve the constrained optimization problem
\begin{equation}
\label{P_d}
\max_{w \in [0,1/q]^J} \mbox{  }E_T\{ R(w)\} \mbox{\,   subject to  \,}   \sum_{i=1}^{J} w_i = J. 
\end{equation}

It was not noted in these works that this problem is convex. Usually convex programs are about minimization of convex functions. This is equivalent to maximizing concave functions. In our case the objective is a sum of terms of the form $\Phi\{ \Phi^{-1}(qw_i)-\mu_i\}$, whose concavity follows directly by differentiation. Yet, by simple Langrangian optimization, the above papers show that if all $\mu_i <0$, the optimal weights are $w_i = w(\mu_i)$, where
\begin{equation}
w(\mu) = \Phi\left(\frac{\mu}{2}+ \frac{c}{\mu}\right)/q.
\label{Spjotvoll}
\end{equation}

Here $c$ is the unique normalizing constant such that the weights sum to $J$. Interestingly, the weights are not monotonic as a function of $\mu$, but maximal for intermediate values of $\mu$ \citep{roeder2009genome}.

  As noted by \cite{roeder2009genome} the formula is a direct consequence of Spjotvoll's theory of optimality in multiple testing \citep{spjotvoll1972optimality}. Accordingly, we call these the \emph{Spjotvoll weights}.

\subsection{Weighting leads to a substantial power gain}
\label{power_gain_sparse} 
To illustrate theoretically that weighting can lead to increased power, we compare the power of optimal weighting and unweighted testing in a sparse mixture model. 

P-value weighting exploits the heterogeneity of the tests. In the simplest case there are only large and small effects, say $M \ll m \approx 0$. We want the $m \to 0$ limit, and for simplicity we suppose $m= 0$.  Similar results hold if $m \approx 0$. Let the fraction of large and small means be $\pi_1, \pi_0>0$, so that $\pi_1 J$ means are $M$, and the remaining $\pi_0J$ are 0. We solve below for the optimal weights.


\begin{proposition}[Optimal weights for sparse means] There is a set of optimal weights that gives the same weight to equal means, $w_0$ and $w_1$ to $0$ and $M$. These are:
$$
(w_0,w_1)=
\left\{
	\begin{array}{ll}
		(0,1/\pi_1)  & \mbox{\, if \, }  \pi_1 \Phi(-|M|/2) > q, \\
		\left\{ \frac{q - \pi_1\Phi(-|M|/2)}{q \pi_0} ,  \frac{\Phi(-|M|/2)}{q} \right\} & \mbox{\, if \, } \pi_1 \Phi(-|M|/2) \le q.
	\end{array}
\right.
$$

Further, the power $p^*$ of the optimal p-value weighting method is:
$$
p^*(\pi_1,M,q)=
\left\{
	\begin{array}{ll}
		\pi_1\Phi \left\{ \Phi^{-1} (q/\pi_1) + |M| \right\} & \mbox{\, if \, }  \pi_1 \Phi(-|M|/2) > q, \\
		q + \pi_1\left\{ \Phi (|M|/2) - \Phi(-|M|/2) \right\} &\mbox{\, if \, }  \pi_1 \Phi(-|M|/2) \le q.
	\end{array}
\right.
$$

\label{Opt_Sparse}
\end{proposition}

If $|M|$ is small enough that $\pi_1 \Phi(-|M|/2) > q$, all the weight is placed on the larger means. This is the behavior we expect intuitively. However, if $|M|$ is large enough that $\pi_1 \Phi(-|M|/2)\le q$, then it is advantageous to place some weight on the small means. The reason is that such a large $|M|$ will be detected with very high probability.

The power of unweighted Bonferroni is $
p_\mathrm{unif}(\pi_1,M,q) = \pi_0 q + \pi_1 \Phi\{\Phi^{-1}(q) + |M|\}. 
$

\begin{figure}
\centering
\begin{subfigure}{.5\textwidth}
\centering
\includegraphics[scale=1]{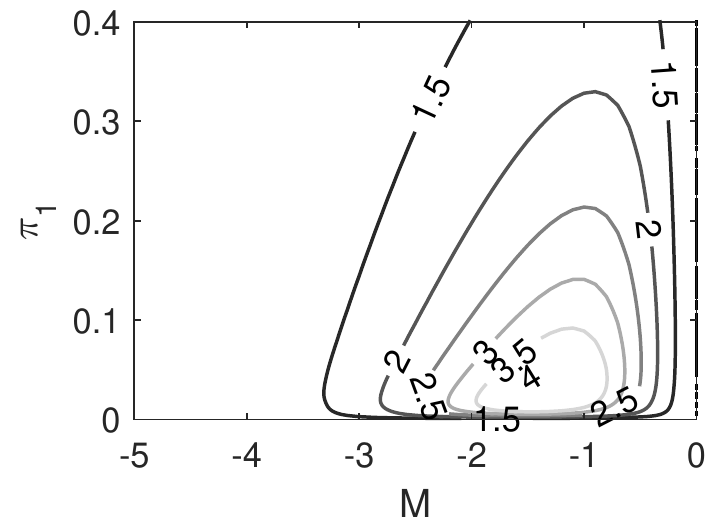}
\caption{Power gain.}
\end{subfigure}%
\begin{subfigure}{.5\textwidth}
  \centering
  \includegraphics[scale=1]{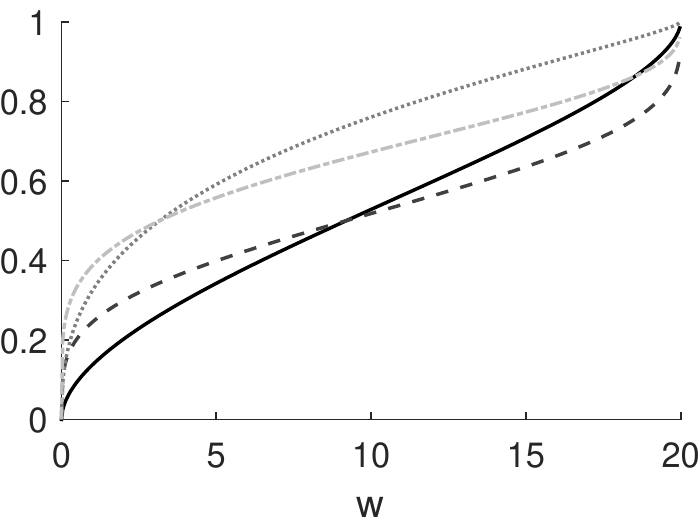}
  \caption{Non-convex summand.}
\end{subfigure}
\caption{(a) Contour plot of power ratio of optimal and unweighted testing for sparse means. (b) Plots of four different instances of the function that is summed in the optimization objective. The non-convex function $w \to \Phi[\{\Phi^{-1}(qw)-\eta\}/\gamma]$ with $q = 0.05$ is plotted for the following pairs $(\eta,\sigma)$: $(-0\cdot1,1)$ solid,  $(-0\cdot1,2)$ dashed, $(-1,1)$ dotted, $(-1,2)$ dot-dashed.}
\label{Power_Sparse}
\end{figure}

In Figure \ref{Power_Sparse} (a), we plot the ratio of $p^*(\pi_1,M,q)/ p_\mathrm{unif}(\pi_1,M,q)$ for $q = 10^{-3}$ for a range of $M$ and $\pi_1$. We observe that for most effect sizes  $M \in [-2.5,-0.25]$, and $\pi_1 < 0.4$, we get a power boost of at least $50\%$. Further, there is a hotspot where the power gain can be 3-4 fold. 
Thus, optimal weighting can lead to a significant boost in power. 

\subsection{Weights with imperfect prior knowledge}
\label{main_theory}

In the previous sections it was assumed that the effects $\mu_i$ were known precisely. We will instead assume that we have uncertain prior information $\mu_i \sim \mathcal{N}(\eta_i,\sigma_i^2)$ about them. 

We maximize the expected power $E_{\mu}[ E_{T}\{ R(w)\}]$. The expectation is with respect to the random $T_i$ and $\mu_i$. Introducing $\gamma_i = (\sigma_i^2+1)^{1/2}$, the \emph{Bayes weights} problem becomes: 
\begin{equation}
\label{P_avg}
\max_{w\in [0,1/q]^J} \mbox{  } \sum_{i=1}^{J}\Phi\left\{\frac{\Phi^{-1}(qw_i)-\eta_i}{\gamma_i}\right\}
 \mbox{\,   subject to  \,}   \sum_{i=1}^{J} w_i = J. 
\end{equation}

This objective is not concave if any $\gamma_i >1$. To help with visualization, the function $w \to \Phi[\{\Phi^{-1}(qw)-\eta\}/\gamma]$ is plotted in Figure \ref{Power_Sparse} (b) for four parameter pairs $(\eta,\gamma)$. The function is increasing, and its curvature has two intervals: concavity, followed by convexity.

Our key contribution is to solve this problem efficiently for large $J$. Our results here are twofold. First, we can solve the problem exactly in the special case when $q$ is sufficiently small. Second, we have a nearly exact solution for arbitrary $q$. We start with the simpler first case. Let us define
\begin{equation}
\label{root}
c(\eta,\gamma; \lambda) = - \frac{\eta + \gamma \{
\eta^2 + 2\left(\gamma^2-1\right)\log(\gamma\lambda)\}^{1/2}}{\gamma^2-1}.
\end{equation}

It turns out that $c$ is the optimal critical value when
\begin{equation}
\label{dual_variable_existence}
   q \le \frac 1J \sum_{i=1}^{J} \Phi\left\{c(\eta_i,\gamma_i;1) \right\}.\\ 
\end{equation}

In our data analysis examples and simulations, this upper bound requires that $q$ be below values in the range $0\cdot 1 - 0\cdot 3$. In the next result we find the exact optimal weights for small $q$ when all $\sigma_i>0$.

\begin{theorem}[Form of Bayes weights] If the significance level $q \in [0,1]$ is small enough that \eqref{dual_variable_existence} holds then the optimal \emph{Bayes weights} maximizing the average power \eqref{P_avg} are $
w_i = w(\eta_i,\gamma_i; \lambda) = \Phi\{c(\eta_i,\gamma_i; \lambda)\}/q
$, where $\lambda \ge 1 $ is the unique constant such that $\sum_{i=1}^{J} w(\eta_i, \gamma_i; \lambda) = J$.
\label{Form of Bayes weights}
\end{theorem}

In the supplementary material, we solve this problem by maximizing the Lagrangian. The two key properties are the joint separability of the objective and constraint; and the analytic tractability of the Gaussian. 

Figure \ref{reg_weights} shows surface and contour plots of an instance of the optimal weights $w(\eta,\sigma)$, as a function of the mean $\eta$ and standard deviation $\sigma$. In the theorem the weights are a function of $(\eta,\gamma)$. Here and below, we will often view them as a function of $(\eta,\sigma)$, via the natural map $\gamma^2 = \sigma^2+1$. As the standard error $\sigma$ becomes small, our weights tend to the Spjotvoll weights: 

\begin{figure}
\centering
\begin{subfigure}{.5\textwidth}
  \centering
  \includegraphics[scale=1.2]{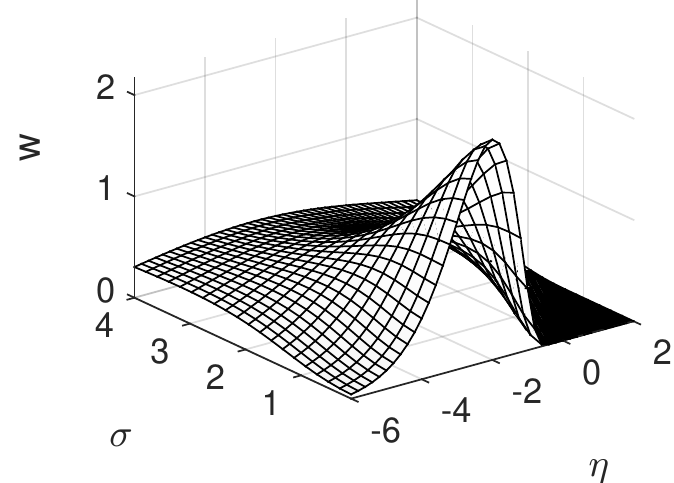}
  \caption{Surface}
\end{subfigure}%
\begin{subfigure}{.5\textwidth}
  \centering
  \includegraphics[scale=1.2]{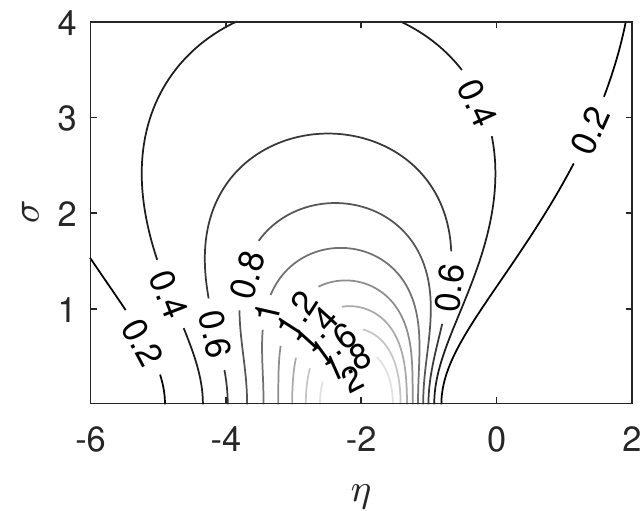}
  \caption{Contour}
\end{subfigure}
\caption{Bayes weights: (a) surface and (b) contour plots of the Bayes weight function $w(\eta,\sigma)$ defined in Theorem \ref{Form of Bayes weights}. Spjotvoll weights are on the segment $\sigma=0$, $\eta<0$.}
\label{reg_weights}
\end{figure}

\begin{proposition}[Recovering the Spjotvoll weights]
\label{recover_spjotvoll}
For any $\lambda$ and $\eta<0$, the Bayes weight function defined by $w(\eta,\gamma; \lambda) = \Phi\{c(\eta,\gamma; \lambda)\}/q$ tends to the Spjotvoll weight defined in \eqref{Spjotvoll} as $\sigma\to0$.
\end{proposition}


With $\sigma_i > 0$, the weights are regularized: more extreme weights are  shrunk towards a common value in a nonlinear way. For finite $\sigma_i$ our weights can be viewed as a smooth interpolation between Spjotvoll and uniform weights. It is reasonable to think at first that as all $\sigma_i \to \infty$, the best weight allocation is uniform. This is not the case. As we will show below, an interesting symmetry breaking phenomenon occurs.

Consider a weight vector $w$ that equals $1/q$ for $\left \lfloor{Jq}\right \rfloor$ indices, and assume $Jq$ is not an integer. Distribute the remaining strictly positive weight equally among the remaining hypotheses. Now it is easy to see that for the $1/q$ weights we always reject, thus the power equals 1. For the remaining weights the objective $\Phi[\{\Phi^{-1}(qw)-\eta\}/\gamma]$ tends to $\Phi(0)=1/2$ as $\sigma \to \infty$ (i.e. $\gamma \to \infty$). This shows that the limiting power is $\left \lfloor{Jq}\right \rfloor + (J-\left \lfloor{Jq}\right \rfloor)/2 = (J +  \lfloor {Jq} \rfloor)/2$. This is larger than $J/2$, the limit power of uniform weighting! This illustrates the symmetry breaking phenomenon caused by the extreme non-convexity of the optimization problem.

Fortunately the situation is better as long as condition \eqref{dual_variable_existence} holds. This condition is easy to check for any given parameters $q$ and $(\eta_i,\sigma_i)$, $ i = 1, \ldots, J$.  In addition we will now show theoretically that the constraint is mild.  Often, even if $J$ is large, we want to keep $\alpha = Jq$ small, because $\alpha$ is the number of false rejections we tolerate.  In this regime the condition holds as long as there are a few average-sized negative prior means $\eta_i$ (denote $z_{c} = \Phi^{-1}(c)$).

\begin{proposition}[Simple condition]
\label{simple_condition_thm}
Condition \eqref{dual_variable_existence} is true (and so Theorem \ref{Form of Bayes weights} applies) if there are $K$ distinct indices $i$ with negative $\eta_i$, such that $
\gamma_i^2\log(\gamma_i^2)/|z_{\alpha/K}| \le |\eta_i| \le |z_{\alpha/K}|.
$
\end{proposition}


If $\alpha/K \to 0$, it is well-known that asymptotically $|z_{\alpha/K}| \sim \{2 \log(K/\alpha)\}^{1/2}$, so the simple condition holds as long as:
$ \gamma_i^2\log(\gamma_i^2) \{2 \log(K/\alpha)\}^{-1/2} \lesssim |\eta_i| \lesssim \{2 \log(K/\alpha)\}^{1/2}.$ This is a quite weak requirement. 
For instance, if $K=10$, $\alpha=0\cdot 01$, then $\{2 \log(K/\alpha)\}^{1/2}$ approximately equals $3.7$. Continuing this example, if $\sigma=1$, so that $\gamma^2 = 2$, and $\gamma^2  \log (\gamma^2)/3.7 = 0.16$, then we only need 10 hypotheses with $0.16 \le |\eta| \le 3.7$.

When $q$ is small, we use Newton's method to find the right constant $\lambda$ from the theorem via a one-dimensional line search. 
The function evaluations cost $O(J)$ per iteration, and empirically it takes a small number of iterations independent of $J$ to converge. In our data analysis section, we solve problems with more than 2 million tests in a few seconds on a desktop computer. 

Now we move to presenting our result for the general case.

\begin{theorem}[Weights in the general case]
For any $q \in [0,1]$, the non-convex Bayes weights problem can be solved for a nearby $q^* \in [0,1]$, for which $ |q^*-q| \le 1/(2J)$.
The optimal weights and $q^*$ can be found in $O(J \log J)$ steps.
\label{complete_nonconvex}
\end{theorem}


This result is most relevant for the settings when $\alpha=Jq$, the expected number of errors under the null, is set to at least 1/2. If so, and especially for large $J$, our weights will be optimal for a $q^*$ that is close to $q$. We see from the proof that even for large $q$, $q^*$ often equals $q$. The method also returns the value $q^*$, which the user can inspect. It is then up to the user to choose whether to perform multiple testing adjustment at the original level $q$ or at the new level $q^*$.

We note that the analysis of non-convex optimization problems is challenging in general. It is perhaps remarkable that the non-convex Bayes weights problem admits a nearly exact solution. 

\section{Simulation studies}
\label{sec:sim}

\subsection{Bayes weights are more powerful than competing methods}
\label{sec:compare_methods}

\begin{figure}
  \centering
  \includegraphics[scale=1]{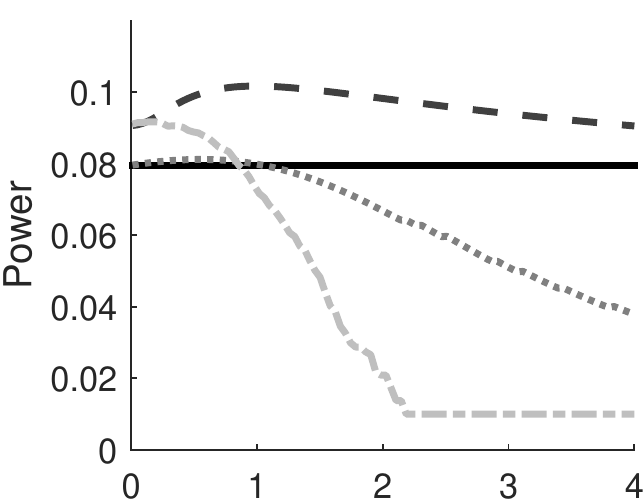}
\caption{Power of four p-value weighting methods as a function of their parameter. Unweighted (solid), Bayes (dashed) as a function of the dispersion $\phi$, exponential (dotted) as a function of $\beta$, and filtering (dot-dashed) as a function of $|M|$. The Spjotvoll weights correspond to the point at the origin $\phi=0$ on the Bayes weights curve. }
\label{fig:power_comparison}
\end{figure}

We present two simulation studies to explore the empirical performance of our method. First we show that Bayes weights increase power more reliably than competing methods.  We compare three methods of p-value weighting: Bayes, exponential, and filtering.

For Bayes weights we multiply the variances by a dispersion factor $\phi$: $\mathcal{N}(\eta_i, \phi\sigma_i^2)$. By changing the dispersion, we test the robustness of our method to mis-specification of the prior variances. This corresponds to the same dispersion variable $\phi$ in our GWAS application given in the next section. The dispersion ranges from 0 to 4. Spjotvoll weights correspond to $\phi= 0$.

Exponential weights with tilt $\beta$ are defined as: $w_i = \exp(\beta|\eta_i|)/c$, where $c = \sum_{i=1}^{J}  \exp(\beta |\eta_i|)$. This weighting scheme was proposed in \citep{roeder2006using}, who recommend $\beta = 2$ as a default. We include the range $\beta \in [0,4]$. As noted by \cite{roeder2006using}, exponential weights are sensitive to large means. To guard against this sensitivity, we truncate the weights larger than $1/q$ and re-distribute their excess weight among the next largest weights. 

Filtering methods test only the most significant effects $\eta_i \le M$, and give them equal weights. This and related methods are known under many names, such as two-stage testing, screening, or the proxy phenotype method \citep{rietveld2014common}. We adopt the name filtering from \cite{bourgon2010independent}, who filter based on independent information in the current data set instead of prior information.
The threshold $M$ ranges from $-4$ to 0. If $|M|$ is large and fewer than $Jq$ hypotheses would be tested, then we instead test the most significant $Jq$.

The simulation is conducted as follows: We generate $J=1000$ random means and variances independently according to $\eta_i \sim \mathcal{N}(0,1)$, $\sigma_i \sim |\mathcal{N}(0,1)|$. We set $q = 10^{-2}$. For any weight vector $w$, we calculate the power as the objective from \eqref{P_avg} divided by $J$, to reflect the average power per test.

The results are shown in Figure \ref{fig:power_comparison}. Each method can improve the power over unweighted testing. However, Bayes weights lead to more power than the other methods. The best power is attained when the dispersion $\phi=1$, but good power is reached even when $\phi$ is not 1. Our weights are robust to mis-specifying the dispersion. 

Note, in particular, that taking uncertainty into account helps. Spjotvoll weights, which assume fixed and known effects, and are shown on the figure as regularized weights with $\phi = 0$, have less power than Bayes weights with positive $\phi$, for a wide range of $\phi$.

The remaining two methods, filtering and exponential weights, have disadvantages. While filtering leads to power gain for a thresholding parameter $M \lesssim 3/4$, it also leads to a substantial power loss for $M>1$. For sufficiently large $M$ the power equals $q$, because only the top $Jq$ hypotheses are selected. Further, it is a significant disadvantage that there is no principled way to choose $M$ a priori without additional assumptions.

Similarly, exponential weighting leads to at most a small gain in power, and usually leads to a loss. There also appears to be no simple, principled way to choose $\beta$ a priori.

We conclude that Bayes weights are quite insensitive to tuning and have uniformly good power. In contrast, exponential weighting and filtering are relatively sensitive and their power can drop substantially. Therefore, Bayes weights increase power more reliably than competing methods.

\subsection{Bayes weights have a worst-case advantage}
\label{sim_robust}

We show that Bayes weights have a worst-case advantage compared to Spjotvoll weights. We use the sparse means model: we  generate $J = 1000$ means $\eta_i$, such that their distribution is $\eta \sim \pi_0 \delta_m + \pi_1 \delta_M$, where $m = -10^{-3}$ and $M = -2$. We set $q = 10^{-2}$ and we vary $\pi_1$ from 0 to 0.1. We set all $\sigma_i=\sigma$.

We consider two values of $\sigma$: 0 and 1. Spjotvoll is optimal for $\sigma=0$, while Bayes weights with $\sigma=1$ are optimal for 1. We evaluate these weighting schemes by calculating the objective \emph{that they do not maximize}: 
the average power \eqref{P_avg} for Spjotvoll and the deterministic power \eqref{P_d} for Bayes weighting. We also compute the power of the unweighted Bonferroni method. 

\begin{figure}
\centering
\begin{subfigure}{.5\textwidth}
 \centering
\includegraphics[scale=1.2]{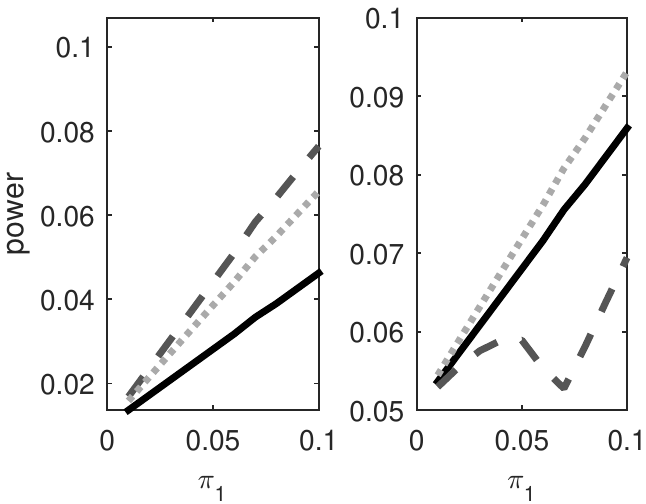}
\caption{Power comparison for sparse means.}

\end{subfigure}%
\begin{subfigure}{.5\textwidth}
\centering
\includegraphics[scale=1.2]{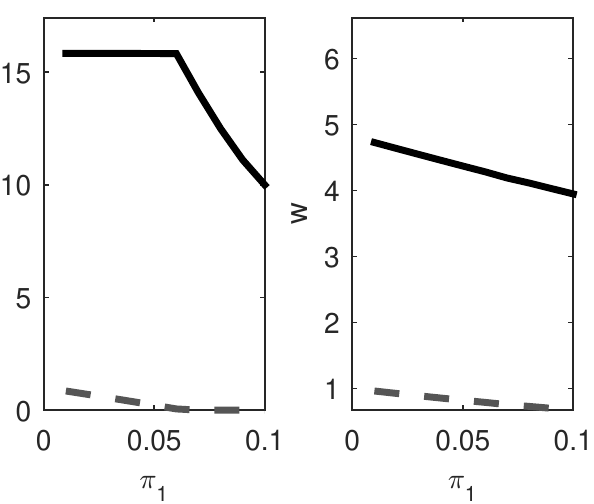}
\caption{Weights of the two classes.}
\end{subfigure}
\caption{(a) Deterministic (left) and average (right) power as a function of the proportion of large means $\pi_1$. Three methods are compared: unweighted (solid), Spjotvoll (dashed), Bayes (dotted). (b) The weights of the two classes, large (solid) and small (dashed), for Spjotvoll (left) and Bayes (right), as a function of the proportion of large means $\pi_1$. }
\label{pow_sparse}
\end{figure}

The results are in Figure \ref{pow_sparse}. Bayes weights lose only a little compared to the optimal Spjotvoll ((a) left). In contrast, Spjotvoll loses a lot compared to Bayes ((a) right). Bayes maximizes the worst-case power, thus showing a maximin property.

Spjotvoll weights show a marked drop in power near $\pi_1 = 0.07$, as shown by its non-monotonic power curve in Figure \ref{pow_sparse} (a)). To understand this, we plot the two weighting schemes in Figure \ref{pow_sparse} (b). Since there are only two classes, the weights also take two values. We see that the Spjotvoll weights are more extreme than the Bayes ones. They start putting weight \emph{equal to zero} on the small means near $\pi_1=0.07$, which appears to lead to power loss.

\section{Application to Genome-Wide Association Studies}
\label{sec:igwas}

\subsection{Review of GWAS}
We adapt our framework to genome-wide association studies, relying on basic notions of quantitative genetics \citep[see e.g.][]{lynch1998genetics}. Our method applied to this problem is called iGWAS in our forthcoming application to human longevity. This section presents in detail the methodology for that application, while also illustrating the steps to use our framework for specific problems. 

Consider a model for GWAS in which we study a quantitative trait $y$ in a population. Our goal is to understand the effects of single nucleotide polymorphisms (SNPs) $g_1, g_2, \ldots, g_J$ on the trait. We assume $y$ has mean 0 and known variance, and $g_i$ denotes the centered minor allele count of variant $i$ for an individual. We rely on the linear model for the effect of the $i$-th variant on the trait:  $y = g_i \beta_i + \varepsilon_i$.

We will show how our question can be framed in the Gaussian means model of hypothesis testing. In the above model $y$ is the phenotype of a randomly sampled individual from the population. Accordingly, $g_i$ is random, $\beta_i$ is a fixed unknown constant, and $\varepsilon_i$ is the residual error. This error is a mean-zero random variable independent of $g_i$, with variance $\sigma_i^2$. 

Suppose we observe a sample of $N$ independent and identically distributed observations from this model. We use the standard linear regression estimate $\hat{\beta_i}$, which for large sample size has an approximate distribution $N^{1/2} \hat{ \beta_i} \stackrel{.}{\sim} \mathcal{N}\{N^{1/2}\beta_i , \sigma_i^2/\mathrm{var}(g_i)\}$.  We can standardize if we divide by $\tau_i$, where $\tau_i^2 = \sigma_i^2/\mathrm{var}(g_i)$ is the variance of $N^{1/2}\hat{\beta_i}$.

With these steps, we have framed our problem in the Gaussian means model. Denoting $T_i = N^{1/2} \hat{\beta_i}/\tau_i$, $ \mu_i =N^{1/2} \beta_i/\tau_i$, we have $T_i \stackrel{.}{\sim}\mathcal{N}(\mu_i,1)$, which is the required form. Denote also the standardized effect size  $\nu_i =\beta_i/\tau_i$, which will be of key importance.

\subsection{Prior Information}

Now we show how to use prior information. Assume we also have a \emph{prior} trait $y_0$ measured independently on a different, independent sample from the same population. If our assumptions also hold for $y_0$, we can write $y_0 = g_i \beta_{0i} + \varepsilon_{0i}$. Here $\beta_{0i}$ is a fixed unknown constant, and $\varepsilon_{0i}$ is random. Suppose we have independent samples of size $N_i$ and $N_{0i}$ for the two traits. If we define $T_{0i},\nu_{0i}$ by analogy to their definitions for $y$, we can write $T_{0i} \sim \mathcal{N}(N_{0i}^{1/2}\nu_{0i}, 1)$.

We model the relatedness between the two traits as a relation between the standardized effect sizes $\nu$, which do not depend on the sample size. If the two traits are closely related, the first order approximation is equality: $\nu_i = \nu_{0i}$. This simple model captures the pleiotropy between the two traits \citep[e.g. ][]{solovieff2013pleiotropy}. 

The final step is to compute the distribution of $\nu_i$ given the prior data $T_{0i}$. For this we need to choose a prior for $\nu_i$, and for simplicity we will use a flat prior. 

We now have all ingredients to apply the model for Gaussian hypothesis testing with uncertain information. Specifically, we have $\mu_i \sim \mathcal{N}(\eta_i,\sigma_i^2)$, where $\mu_i = N_i^{1/2}\nu_i$, $\eta_i  = (N_i/N_{0i})^{1/2}T_{0i}$, and $\sigma_i^2 = N_i/N_{0i}$. 

The uncertainty in $T_{0i}$ may be different from 1, for instance larger than 1 due to overdispersion. In addition, overdispersion is one way to weaken the first order approximation assumption. To allow for this, we recall the dispersion parameter $\phi$ used in our simulation. We model the prior data as $T_{0i} \sim \mathcal{N}(N_{0i}^{1/2}\nu_{0i}, \phi)$. Then the variance  $\sigma_i^2 = \phi N_i/N_{0i}$. The default is $\phi = 1$. Finally, we compute Bayes weights $w_i$ with parameters $q$, $(\eta_i,\sigma_i^2)$, and run weighted Bonferroni on the current p-values. This fully specifies the method. For the reader's convenience, the method is summarized in Algorithm \eqref{iGWAS}.


\begin{algorithm}
\caption{Multiple Testing with Bayes weights in GWAS}\label{iGWAS}
\vspace*{-12pt}
\begin{tabbing}
   \enspace $T_{0i} \gets $ prior effect sizes for $i=1,\ldots,J$\\
   \enspace $N_{0i}, N_i \gets $ prior and current sample sizes \\
   \enspace $P_i \gets $ current p-values \\
   \enspace $q \gets $ significance threshold \\
   \enspace $\phi \gets$ dispersion (default $\phi=1$) \\
   \enspace Set the prior means and variances: $\eta_i  = (N_i/N_{0i})^{1/2}T_{0i}$, \, $\sigma_i^2 = \phi N_i/N_{0i}$ \\
 \enspace Compute Bayes weights $w_i$, defined via \eqref{P_avg},  with parameters $q$, $(\eta_i,\sigma_i^2)$\\
  \enspace Return indices $i$ such that $P_i \le qw_i$\\

\end{tabbing}
\end{algorithm}

\subsection{Practical remarks}
It is important that we retain type I error control as soon as we have valid p-values, even if the modelling assumptions fail.  Common deviations from our model are: \textbf{(1)} GWAS summary data sometimes only has the magnitude of the effects, and not their sign. In this case we have two choices. We may assume that the directions of effects are the same, and do a one-tailed test of the current effect in the prior direction. Alternatively, we can do a two-tailed test by including the 2 tests $(\eta_i,\sigma_i^2)$ and $(-\eta_i,\sigma_i^2)$ for each $i$, for a total of $2J$ tests. \textbf{(2)} When the prior and current trait are not both quantitative, and one is binary, the model $\nu_i=\nu_{0i}$ should be re-examined. It is still convenient as a first approximation. \textbf{(3)} Benjamini-Hochberg may be used instead of Bonferroni, with any weights summing to $J$, for increased power \citep{genovese2006false}.
 
\section{Data Analysis}
\label{sec:data_analysis}

\subsection{Data Sources}
\label{detailed_data}

We illustrate our method by analyzing 5 publicly available genome-wide association studies of quantitative or binary traits. We use the association p-values, which are available for 500,000 to 2.5 million SNPs. The five studies are: CARDIoGRAM and C4D for coronary artery disease \citep{schunkert2011large, coronary2011genome}, and one each for the kidney trait estimated glomerular filtration rate (eGFR) creatinine \citep{kottgen2010new}, blood lipids  \citep{teslovich2010biological}, and schizophrenia \citep{schizophrenia2011genome}. CARDIoGRAM and C4D include non-overlapping samples. A detailed description appears in the Supplementary material.

\subsection{Specific Pairs}
\label{specific+pairs}

We analyze three pairs of data sets, with specific motivation for each. First, we use CARDIoGRAM as prior information for C4D. This is a `positive control' for our method, since both studies measure the same phenotype, coronary artery disease. Therefore the weights should increase power. We choose C4D as target because it has smaller sample size; hence prior information may increase power more substantially.

Second, we use the blood lipids GWAS as prior information for schizophrenia. \cite{andreassen2013improved} showed improved power with this pair. They developed and controlled the Bayesian conditional false discovery rate. This is not known to control a frequentist criterion. Our goal was to evaluate the power improvement using a frequentist method. As \cite{andreassen2013improved} noted, there is a small overlap between the controls of the two GWAS (Section \ref{detailed_data}). 

Third, we used the eGFR creatinine GWAS as prior information for the C4D coronary artery disease study. Heart disease and renal disease are comorbid \cite[eg.][]{silverberg2004association}, so it is possible that this may improve power. Here the hypothesized improvement is not based on entirely rigorous arguments.

\subsection{Methods compared}
\label{Methods_compared}
We run weighted Bonferroni multiple testing for each of the 5 weighting schemes. The prior data is $T_{0i}  = \Phi^{-1}(P_{0i}/2)$, where $P_{0i}$ is the $i$-th prior p-value. The family-wise error rate is controlled at $q=0.05$. 

The first four methods are: unweighted Bonferroni, where all weights equal 1; Spjotvoll weights with parameters $\mu_i  =(N_i /N_{0i})^{1/2}T_{0i}$;  Bayes weights defined in Section \ref{sec:igwas}, with dispersion $\phi = 0.1, 1$, and  $10$; and exponential weighting \citep{roeder2006using} with tilt $\beta = 1,2$, and 4, introduced in Section \ref{sec:compare_methods}.

The fifth and last method is filtering, which selects the smallest p-values in the prior study, and tests their SNPs in the current study.  We use three p-value thresholds $P < 10^{-2},  10^{-4},  10^{-6}$. \cite{rietveld2014common} propose a method to choose the optimal p-value threshold for filtering. This needs the genotypic correlation between the two traits and the additive heritability of the current trait. For complex traits these parameters are usually estimated with a large uncertainty. Substantial domain expertise is required to choose the right parameter. 

\subsection{Additional details}

We prune the significant SNPs for linkage disequilibrium (LD) using the DistiLD database \citep{palleja2012distild}. We LD prune the significant SNPs for each method by selecting one SNP from each LD block.  Our data analysis pipeline is available from the first author.

 We compute a score $s_{m(p)d}$ for each method $m$ with parameters $p$, on each data set $d$. This is defined as +1 if the method increases the number of hits compared to unweighted, 0 if it leaves it unchanged, and -1 otherwise. The score $s_{m(p)}$ of a method $m$ with parameters $p$ is the sum of scores across data sets. The total $s_m$ of the method $m$ is the sum of scores $s_{m(p)}$ across parameters.

\subsection{Results}
\begin{table}[h]
\caption{Number of significant SNPs of five methods on three examples. Top: Results pruned for linkage disequilibrium (LD). Middle: results without LD pruning. Bottom: The score of each method. The methods compared are unweighted (Un); Spjotvoll (Spjot); Bayes with $\phi=0\cdot1,1,10$; exponential (Exp) with $\beta=1,2,4$ and filtering (Filter) with $-\log(P)=2,4,6$. CG stands for CARDIoGRAM}
\begin{tabular}{rrr@{\hskip 0.35in}rrr@{\hskip 0.35in}rrr@{\hskip 0.35in}rrr}
                     & Un & Spjot & \multicolumn{3}{c}{Bayes($\phi$)} & \multicolumn{3}{c}{Exp($\beta$)} & \multicolumn{3}{c}{Filter($-\log P$)} \\
Parameter                     &    &       & $0\cdot 1$       & 1       & 10       & 1            & 2            & 4           & 2  &4  & 6 \\
\\  
\multicolumn{12}{l}{\textbf{LD Pruned}}      
                                                                                                                   \\
CG $\to$ C4D &4    & 11      &  10                &  8       &    4      &          4    &    5          &  4           &     10       &    10        &  6         \\
Lipids  $\to$ SCZ    &  4  &   1    &     1   & 1          &     5    &     1     &       0       &       0       &    2         &       2     &      2                \\
eGFRcrea $\to$ C4D    & 4   &  2     &    2              &   4      &    4      &      4        &    5          &       4      &        1    &    0        &   1        \\
\\  
\multicolumn{12}{l}{\textbf{Unpruned}}                                                                                                                               \\
CG$\to$ C4D  & 29   &   45    &        44          &   39      &   29       &     32         &   34           &   27          &   40         &    48        &    34      \\
Lipids  $\to$ SCZ     & 116   &    214   &    214              &   223      &    123      &     92         &    0          &      0       &    217        &   96         &    39       \\
eGFRcrea $\to$ C4D          & 29   &   18    &   18               &     23    &    29      &     29         &   28           &   19          &   1         & 0           &    1       \\
\\  
\textbf{Scoring}              &    &       &                  &         &          &              &              &             &            &            &           \\
Score                & 0   &  0     &      0            &  1       &        1  &      0        &    0         &   $-$1           &     0       &     $-$1        &     $-$1       \\
Total                &  0  &   0    &    \multicolumn{3}{c}{sum  =  2}          &        \multicolumn{3}{c}{sum  =  $-$1}         &      \multicolumn{3}{c}{sum  =  $-$2}           \\
                     &    &       &                  &         &          &              &              &             &            &            &          
\end{tabular}
\label{data_analysis_results}
\end{table}
The results of our data analysis are presented in Table \ref{data_analysis_results}. This table has the number of significant SNPs on the pairs of GWAS data in \ref{specific+pairs} using the weights in \ref{Methods_compared}. We also present the LD pruned results, which are a proxy for the the number of independent loci found.

The results are somewhat inconclusive. On the positive control example, all methods except exponential weighting improve power. Spjotvoll weighting and filtering have the largest number of SNPs. On the blood lipids example, methods generally lose power for LD pruned SNPs (except Bayes weights with $\phi=10$); and methods can both lose and gain power for non-LD pruned SNPs (except Bayes weights which uniformly improves power). On the other hand, for the eGFR creatinine example, exponential weights show the best behavior. We also see that the default $\phi=1$ is never worse that both unweighted and Spjotvoll, and for the unpruned lipids example it is better. 

If we allow for tuning of parameters, Bayes weights show a good performance. They are either first or second in all examples, and other methods rank lower. However, since we don't have a principled way to tune the parameters, we do not pursue this way of evaluation.

Instead, we look at the scoring method for evaluation. The only method with a positive score is our Bayes weights ($\phi=1$ and 10). The total score, summed across parameter settings, is also only positive for Bayes testing. This shows some promise for our method. However, from this analysis alone we cannot establish conclusively the relative merits of the methods. In future work it will be necessary to evaluate p-value weighting methods on more data sets.

\section{Acknowledgements}

This work was supported by grants from the AFAR/EMF and the NIH/NIA (AG025941), and NSF grant DMS-1407397. The first author thanks David Donoho for support.

\section*{Supplementary material}
\label{SM}

The supplementary material is organized as follows: Accessing the software implementation of our method is described in Section \ref{software}. In Section \ref{proofs} we give the proofs of all mathematical claims from the paper. In addition, in Subsection \ref{method_for_computing_weights} we give Algorithm \ref{weights_alg} to compute our weights. In Section \ref{detailed_data} describe our data  sources in detail.

\section{Software}\label{software}

\subsection{End-user software}

We provide R and MATLAB implementations of the methods developed in this paper. They can be obtained from public git repositories or directly from the first author. The Matlab implementation is available from \url{https://github.com/edgardobriban/pvalue_weighting_matlab}. The R implementation is available from \url{https://github.com/edgardobriban/pvalue_weighting_r}. An R package is under development.

\subsection{Reproducibility}
All computational results and analyses of this paper have been performed in a reproducible way. To reproduce the simulation results and figures, we provide the source code in the above-mentioned Matlab package at \url{https://github.com/edgardobriban/pvalue_weighting_matlab}. To reproduce the data analyses, a separate repository has been created, and is publically available at \url{https://bitbucket.org/edgardobriban/pvalue-weighting-gwas}.

\section{Proofs}\label{proofs}

\subsection{Proof of Proposition 1 - Sparse Means}
\label{Opt_Sparse_Pf}

\begin{proof}
The objective function we need to maximize in $w_i$ is

$$
\sum_{i=1}^{J} \Phi\{\Phi^{-1}(qw_i)-\mu_i\}.
$$

We know that out of the $J$ means $\mu_i$, $\pi_0J$ are equal to $0$. For each of these, the summand in the objective simplifies to $qw_i$. For the remaining $\pi_1J$ means, the same convex objective $g(w) = \Phi\{\Phi^{-1}(qw_i)-\mu_i\}$ is summed, and the objective becomes:

$$
q\sum_{i:\mu_i=0}w_i+ \sum_{i:\mu_i=M}g(w_i).
$$

Now, since $g$ is convex, we have  $\sum_{i:\mu_i=M}g(w_i) \le Mg(\bar w)$, where $\bar w$ is the mean of the weights of large effects $\bar w = (J\pi_1)^{-1} (\sum_{i:\mu_i=M}w_i) $. Hence, there is a set of optimal weights that take equal values for equal means. This proves the first claim in the proposition. 

If we call $w_0, w_1$ the weights for $0$ and $M$, the optimization problem takes the form
\begin{equation*}
 \max_{w_0, w_1 \in [0,1/q]} \mbox{  } \pi_0 w_0 + \pi_1  \Phi\left(\Phi^{-1}(qw_i) - M \right) \mbox{   s.t. }   \pi_0 w_0 + \pi_1  w_1 = 1.
\end{equation*}
 
Next, recall that $qw_i \le 1$ and introduce a new set of variables $c_i := \Phi^{-1}(qw_i)$.  If $qw_i=1$, then $c_i$ will take the value $+\infty$ in the extended real number system $\overline{\mathbb{R}}=\mathbb{R}\cup\{\infty\}$. All our calculations respect the rules of the extended number system, so this will not cause any problems. For instance, $\Phi(\infty)$ is defined as 1 by continuity.

Then the optimization problem becomes
\begin{equation*}
 \max_{c_0, c_1 \in \overline{\mathbb{R}}} \mbox{  } \pi_0 \Phi\left(c_0\right) + \pi_1  \Phi\left(c_1 - M \right) \mbox{   s.t. }   \pi_0 \Phi\left(c_0\right) + \pi_1  \Phi\left(c_1 \right) = q.
\end{equation*}

Substituting the second equation we find that we need to maximize
$$
q + \pi_1\left( \Phi\left(c_1-M\right) - \Phi\left(c_1\right)\right)
$$
subject to $\pi_1 \Phi\left(c_1 \right) \le q$. Now it is easy calculus to check that the function $ c \to  \Phi\left(c-M\right) - \Phi\left(c\right)$ is strictly increasing on $(-\infty,M/2]$ and strictly decreasing on $[M/2,\infty)$. Therefore, if $M/2$ is feasible, i.e. $\pi_1 \Phi\left(M/2 \right) \le q$, then the maximum is achieved at $c_1 = M/2$. This leads to the claimed formula for $w_1 = \Phi(c_1)/q$ and the power $p^*$.

If $M/2$ is not feasible, then the maximum is achieved at the largest feasible value $c_1$, defined by $\pi_1\Phi\left(c_1 \right) = q$. This again leads to the claimed formulas.
\end{proof}

\subsection{Proof of Theorem 1 - Optimal weights for small $q$}
\label{Opt_Weights_Pf}

Similarly to Section \ref{Opt_Sparse_Pf}, we introduce the new set of variables $c_i := \Phi^{-1}(qw_i)$. We can equivalently rewrite our problem as
\begin{equation*}
 \max_{c\in \overline{\mathbb{R}}^J} \mbox{  } \sum_{i=1}^{J} \Phi\left(\frac{c_i - \eta_i}{\gamma_i}\right)  \mbox{s.t. }  \sum_{i=1}^{J} \Phi(c_i) = qJ. 
\end{equation*}

Clearly it is enough to find a scalar dual variable $\lambda$ such that we can solve the Lagrangian problem
\begin{equation*}
 \max_{c\in \overline{\mathbb{R}}^J} \mbox{  } \sum_{i=1}^{J} \Phi\left(\frac{c_i - \eta_i}{\gamma_i}\right) - \lambda \left( \sum_{i=1}^{J} \Phi(c_i) - qJ \right)  \mbox{s.t. }  \sum_{i=1}^{J} \Phi(c_i) = qJ. 
\end{equation*}

The strategy is to study this penalized objective for each fixed $\lambda$, find maximizers $c_i(\lambda)$, and then find a suitable $\lambda$ to make the constraint hold.

For this we introduce a function $f(c;\eta,\gamma)$ that is the generic term in the Lagrangian function of the problem:

$$
f(c;\eta,\gamma) = \Phi\left(\frac{c - \eta}{\gamma}\right)  - \lambda\Phi(c).
$$

We find the maximizer of $f$ in the following lemma. 
\begin{lemma} If $\lambda \ge 1$, then the maximum of $f$ is reached at 

\begin{equation}
\label{root}
c_1(\eta,\gamma; \lambda) = - \frac{\eta + \gamma \sqrt{
\strut\eta^2 + 2\left(\gamma^2-1\right)\log(\gamma\lambda)}}{\gamma^2-1}.
\end{equation}
\label{explicit_maximizer}
\end{lemma}

This function was called $c$ in the main paper. Here we call it $c_1$ to distinguish it from the dummy variable $c$. 

\begin{proof}
Denoting the standard normal density by $\varphi$, the derivative of $f$ with respect to $c$ is:

$$
\frac{\partial f(c;\eta,\gamma)}{\partial c} = \frac{1}{\gamma}\varphi\left(\frac{c - \eta}{\gamma}\right)  - \lambda\varphi(c).
$$

Thus  $\partial f/\partial c\ge 0$ if and only if
\begin{equation}
\label{quad}
c^2 - \left(\frac{c-\eta}{\gamma}\right)^2 \ge 2\log(\gamma \lambda).
\end{equation}

By assumption $\gamma > 1$. We have a quadratic inequality for $c$, and the associated quadratic equation has two roots $c_1, c_2$. If $\lambda\ge 1$ both roots are real: $c_1 \le c_2$. Indeed it is easy to check that the smaller root is given in \eqref{root}. Then $f$ is increasing on $(-\infty, c_1]$ and $[c_2,\infty)$, and decreasing on $[c_1,c_2]$. Further, the limit of $f$ at $-\infty$ is clearly 0, while at $\infty$ it is $1-\lambda$. Therefore, if $\lambda \ge 1$,  $c_1$ is a global maximum of $f$.  
\end{proof}

Therefore, if $\lambda \ge 1$, the $i$-th summand in the Lagrangian is maximized at $c_1(\eta_i, \gamma_i, \lambda)$. To solve our problem it is enough to find a $\lambda \ge 1$ such that $
\sum_{i=1}^{J} \Phi(c_1(\eta_i, \gamma_i,\lambda)) = Jq. $

Such a $\lambda$ exists by condition \eqref{dual_variable_existence}, as follows: $c_1$ is a decreasing function of $\lambda$, thus the constraint $\lambda \ge 1$ is equivalent  to the sum of $\Phi(c_1)$ evaluated at $\lambda=1$ to be greater than $Jq$. This is exactly the condition we assumed in \eqref{dual_variable_existence}. This finishes the proof.

\subsection{Proof of Proposition 2  - Recovering the Spjotvoll weights}
\label{zero_noise_limit}

Let $R^2:= \left(\gamma^2-1\right)\left(\eta^2+2\gamma^2\log(\gamma\lambda)\right)$. Then $c_1(\eta,\gamma; \lambda)$ equals after some calculation:
\begin{align*}
c_1 &= -\frac{\eta+\sqrt{\eta^2+R^2}}{\gamma^2-1} = - \frac{R^2}{\sqrt{\eta^2+R^2}-\eta}\frac{1}{\gamma^2-1} 
= -\frac{\eta^2+2\gamma^2\log(\gamma\lambda)}{\sqrt{\eta^2+R^2}-\eta}.
\end{align*}

As $\sigma \to 0$,  we have $\gamma \to 1$, $\gamma^2\log(\gamma\lambda) \to \log(\lambda)$, and $R^2 \to 0$. For $\eta <0$ this shows that the limit of $c_1$ is
\begin{align*}
 -\frac{\eta^2+2\log(\lambda)}{|\eta|-\eta} =  \frac{\eta^2+2\log(\lambda)}{2\eta} .
\end{align*}

Recall that the Spjotvoll weight is
\begin{equation*}
w(\mu) = \Phi\left(\frac{\mu}{2}+ \frac{c}{\mu}\right)/q.
\end{equation*}

This shows that the we recover the Spjotvoll weights in the limit as $\sigma\to0$.

\subsection{Proof of Proposition 3 - Simple sufficient condition}
\label{simple_condition_proof}

We start by showing that for $\lambda=1$
$$
c_1(\eta,\gamma; \lambda) \ge -\frac{\eta^2+\gamma^2\log(\gamma^2)}{2|\eta|}.
$$

Indeed, for negative $\eta<0$ the left hand side equals: 
$$
- \frac{\eta + \gamma \sqrt{
\strut\eta^2 + \left(\gamma^2-1\right)\log(\gamma^2)}}{\gamma^2-1}
= - \frac{\eta^2 + \gamma^2\log(\gamma^2)}{|\eta| + \gamma \sqrt{
\strut\eta^2 + \left(\gamma^2-1\right)\log(\gamma^2)}}
$$

and now the inequality follows immediately, because $\gamma \ge 1$, so 
$$
\gamma \sqrt{
\strut\eta^2 + \left(\gamma^2-1\right)\log(\gamma^2)} \ge \gamma \sqrt{
\eta^2} \ge |\eta|.
$$

Thus the condition \eqref{dual_variable_existence} of the theorem is satisfied if the following more explicit inequality holds.
\begin{equation}
 \sum_{i=1}^{J} \Phi\left(-\frac{\eta_i^2+\gamma_i^2\log(\gamma_i^2)}{2|\eta_i|} \right) \ge Jq. 
 \label{simpler_inequality}
\end{equation}

For this to hold it is clearly sufficient that there are $M$ distinct indices such that (recall $\alpha = Jq$)
\begin{equation}
\Phi\left(-\frac{\eta_i^2+\gamma_i^2\log(\gamma_i^2)}{2|\eta_i|} \right) \ge \alpha/M.
\label{elementwise}
\end{equation}

The above inequality is equivalent to 

$$\frac{|\eta_i|}{2}+\frac{\gamma_i^2\log(\gamma_i^2)}{2|\eta_i|} \le |z_{\alpha/M}|.$$

By assumption there are $M$ indices such that $|\eta_i|/2 \le |z_{\alpha/M}|/2$ and $\gamma_i^2\log(\gamma_i^2) / (2|\eta_i|) \le$ $ |z_{\alpha/M}|/2$. For these indices \eqref{elementwise} is true, so that \eqref{simpler_inequality} holds. This shows that the original condition \eqref{dual_variable_existence} holds, finishing the proof.

\subsection{Proof of Theorem 2 - Optimal weights in the general case}
\label{complete_sol}

The proof builds on Theorem 1, but requires more detailed analysis. There are two parts: understanding the monotonicity of the generic term in the Lagrangian, then using it to find the optimal weights.

\subsubsection{Monotonicity of the generic term}

We saw in Lemma \ref{explicit_maximizer} that if $\lambda\ge1$, $f$ is maximized at $c_1$. We also saw in the proof that there are two cases: the roots of \eqref{quad} are either real or complex. If the roots are complex, then $f$ is increasing and the supremum is at $+\infty$. The two roots are complex when the discriminant of the quadratic equation is negative:
$$
\eta^2 + 2\left(\gamma^2-1\right)\log(\gamma\lambda)  < 0.
$$

This inequality is equivalent to the following bound for $\lambda$, in terms of a new function $l(\eta,\gamma)$:
\begin{equation}
 \lambda < l(\eta,\gamma) := \frac{1}{\gamma} \exp\left(-\frac{\eta^2}{2(\gamma^2-1)}\right).
\label{l_def}
\end{equation}

Therefore if $\lambda  < l(\eta,\gamma)$, then the supremum of $f$ is at $c = +\infty$. Otherwise there are two candidates for the supremum: $c = +\infty$ and $c_1$.  We want to compare the two candidate maxima. 

 Denote the limit of $f$ at infinity by $p_{\infty}(\lambda) = p_{\infty}(\lambda;\eta,\gamma)$ and the value of $f$ at $c_1$, when defined, as $p_1(\lambda)=p_1(\lambda;\eta,\gamma) = f(c_1)$.  Let the difference between the two extrema be $d(\lambda; \eta,\gamma) = p_1(\lambda;\eta,\gamma) - p_{\infty}(\lambda,\eta,\gamma)$. In the following lemma we find the maximizer of $f$ as a function of $\lambda$.
 
 \begin{lemma}
There is a unique value $k(\eta,\gamma)$  in the interval $[l(\eta,\gamma),1]$ such that $
d(k(\eta,\gamma);\eta,\gamma) = 0.$
For $\lambda < k(\eta,\gamma)$, the supremum of $f$ is at $+\infty$, else it is at $c_1$. For the special value $\lambda = k(\eta,\gamma)$ both values are equal.

 \end{lemma}
 
 \begin{proof} 
We have explicitly $p_{\infty}(\lambda) = 1 - \lambda$ and thus $p_{\infty}$ decreases linearly from 1 to 0 on the unit interval $[0,1]$. We also have by definition

$$p_1(\lambda) = \Phi\left(\frac{c_1(\eta,\gamma; \lambda) - \eta}{\gamma}\right)  - \lambda\Phi(c_1(\eta,\gamma; \lambda)).$$

Differentiating this expression with respect to $\lambda$ reveals

$$\frac{\partial p_1}{\partial \lambda}= \frac{1}{\gamma}\varphi\left(\frac{c_1 - \eta}{\gamma}\right) \frac{\partial c_1}{\partial \lambda}  - \left(\lambda\varphi(c_1)\frac{\partial c_1}{\partial \lambda}  + \Phi(c_1)\right).$$

However, by the definition of $c_1$ we have $\varphi((c_1 - \eta)/\gamma)  - \lambda\varphi(c_1) = 0$. Indeed, $c_1$ was defined as one of the extrema of $f$, which leads to the above equation by taking derivatives. Hence the expression above simplifies to: $\partial p_1/\partial \lambda=- \Phi(c_1)$. 

This shows that $p_1$ is decreasing in $\lambda$: the derivative belongs to $(-1,0)$. We also know that $p_1 >0$, because the value $p_1$ is a local maximum of $f$ in $c$, and we have seen that $f$ - as a function of $c$ -  increases from $0$ (a value which it  takes in the limit at $-\infty$) to $f(c_1)$ - and hence $p_1 = f(c_1)$ must necessarily be positive. We conclude that $p_1(1) > 0 = p_{\infty}(1)$

Finally, we note that $p_1(\lambda)$ is well-defined precisely when $c_1(\lambda)$ is. This happens when the expression inside the square root is non-negative, which means $\lambda \ge l(\eta,\gamma)$. 

To summarize

\begin{itemize}
\item $d$ is defined on the interval $\lambda \in [l(\eta,\gamma),1]$. 
\item $d$ is a strictly increasing differentiable function on this interval, because its derivative is $\frac{\partial d}{\partial \lambda}=1 - \Phi(c_1) > 0$.
\item $d(1) > 0.$
\end{itemize}

If we show that $d$ has a unique root $k$, our conclusion will follow. Given the above three statements, it's enough to show that $d(l(\eta,\gamma)) <0$; then the claim follows by the intermediate value theorem. 

To check the condition $d(l(\eta,\gamma)) <0$ we note that  $c_1(l(\eta,\gamma)) = - \frac{\eta}{\gamma^2-1}$, hence $p_1(l) < p_{\infty}(l)$ is equivalent to

$$ 
\Phi\left( - \frac{\eta \gamma}{\gamma^2-1}\right)  -l\Phi\left( - \frac{\eta}{\gamma^2-1}\right) < 1 - l, 
$$

or after some rearrangement $ 
l\Phi\left( \frac{\eta}{\gamma^2-1}\right)  < \Phi\left( \frac{\eta \gamma}{\gamma^2-1}\right) . $

Introducing the variable $ x  =  \frac{\eta}{\gamma^2-1}$ and eliminating $\eta$ by noting $l = \exp (-x^2 (\gamma^2 - 1)/2 )/\gamma$ we obtain that all we need to show is

$$\exp \left( \frac{x^2}{2}\right) \Phi(x) \le \exp \left( \frac{(\gamma x)^2}{2}\right) \Phi(\gamma x) \gamma.$$

If $\eta<0$, then $x<0$, and the above will follow if the function $m(x) = x \exp \left( \frac{x^2}{2}\right) \Phi(x)$ is strictly increasing on $(-\infty, 0]$. Checking this is an elementary calculus exercise. The cases $\eta=0$ and $\eta>0$ are handled similarly. This finishes the proof.

\end{proof}

\subsubsection{Method for computing weights}
\label{method_for_computing_weights}

In the previous section, we saw that for any $i$, if $\lambda < k_i := k(\eta_i,\gamma_i)$, then the supremum of $f(\cdot; \eta_i,\gamma_i)$ occurs at $+\infty$, else it occurs at $c_1(\eta_i,\gamma_i; \lambda)$. For the optimal weights this shows:
$$
w_i(\eta_i,\gamma_i; \lambda)=
\left\{
	\begin{array}{ll}
		\Phi(c_1(\eta_i,\gamma_i; \lambda))/q  & \mbox{if  }\lambda >  k(\eta_i,\gamma_i), \\
		\Phi(c_1(\eta_i,\gamma_i; \lambda))/q  \mbox{ or } 1/q& \mbox{if  }\lambda =  k(\eta_i,\gamma_i),  \\
		1/q & \mbox{if  }\lambda <  k(\eta_i,\gamma_i).
	\end{array}
\right.
$$

We emphasize that the value of $w_i$ can take either of two values for $\lambda =  k(\eta_i,\gamma_i)$. We will unambiguously specify a choice later for each $i$, determined by the constraint on the sum of the weights.

Now all that remains is to search for a suitable $\lambda$ such that the weights sum to $J$: $\sum_i w_i = J$. If we find such a $\lambda$, then by duality the weights $w_i$ will solve our original problem. Let us denote the sum of the weights by $W$:
$$
W(\lambda) = \sum_{i=1}^{J} w_i(\eta_i,\gamma_i; \lambda).
$$

The function $W(\lambda)$ is unambiguously defined for $\lambda \neq k_i$, which is what we will use at first.

Each function $w_i$ is decreasing in $\lambda$: constant on the interval $(0,k_i)$, and decreasing smoothly from $r_i = \Phi(c_1(\eta_i,\gamma_i; k_i))/q$ to 0 on the interval $(k_i, \infty)$. Further, the function $w_i$ has a jump of size $(1-r_i)/q$ at $k_i$.

Therefore $W$ is a decreasing function of $\lambda$ on $[0,\infty)$, going from $J/q$ to 0. Further, if we sort the values $k_i$ such that $k_{(0)} = 0 < k_{(1)} \le \ldots \le k_{(J)} < k_{(J+1)} = \infty$, then $W$ is smooth on the intervals $(k_{(i)},k_{(i+1)})$, and has jumps of size $(1-r_{(i)})/q$ at $k_{(i)}$ for $1 \le i \le J$. 

Hence, for our problem of solving $W(\lambda) = J$, there are two possibilities:
\begin{paragraph}{Case 1} There is an interval $(k_{(i)},k_{(i+1)})$ such that $J$ belongs to the image of this interval under $W$. 
In this case, since the function W is strictly decreasing and continuous on this interval, there will exist a unique $\lambda$ in the interval such that $W(\lambda)=J$. 
In this case we can solve the original problem exactly.
\end{paragraph}

\begin{paragraph}{Case 2} There is no such interval. In this case, the present duality approach is unable to produce an exact solution to the original problem.

However, we can get an approximate solution. Let us consider the values $W(k_i)$. Above we noted that $w_i(k_i)$ can take two possible values $1/q$ and $r_i$. Hence $W$ can also take more than one value at $k_i$. Since several $k_i$ may be equal, $W$ may take more than two possible values, because each summand $w_i$ can be chosen in two ways. To understand this, let us call the distinct values of $k_i$ to be $K_j$, and assume $K_1 < K_2 < \ldots < K_L$. For brevity let us also add $K_0 = 0$, $K_{L+1} = \infty$. Without loss of generality, we  assume that the values $k_i$ cluster in the following way: $K_1  = k_{(1)}  = k_{(2)}  = \ldots  = k_{(i_1)}$, $K_2  = k_{(i_1 + 1)}   = \ldots  = k_{(i_2)}$, and so on. We  also define the sets $S_i$ containing the indices $j$ for which $k_j = K_i$.

The value of $W$ is now defined unambiguously on each interval $(K_i, K_{i+1})$, and $W$ is a smooth decreasing function on this interval. At $K_m$, $W$ has a jump of size

$$
j_m = \sum_{i: k_i = K_m} \frac{1-r_{(i)}}{q} = \frac{ i_{m+1} - i_m  - \sum_{i=i_m+1}^{i_{m+1}} r_{(i)}}{q}.
$$

By assumption, in the current case there is no interval $(K_i, K_{i+1})$ whose image under $W$ contains $J$. Therefore $J$ is contained in one of the jumps occurring between these intervals, say the jump at $K_i$. Let us now choose $\lambda = K_i$. Let us also define the left and right limits of $W$ at $K_i$: $
a_{+}  = \lim_{x \to K_i, x < K_i} W(x)$, $
a_{-}  = \lim_{x \to K_{i}, x > K_i} W(x). $

Then $a_{-} < a_{+}$ and the current case entails that $ J \in [a_{-}, a_{+}]$. Let us also denote the values of the jumps by the more compact notation $s_i = (1-r_{(i)})/q$. With these notations, we can now easily state that the possible values of $W(K_i)$ are, for any set $S \subset S_i$

$$
W(K_i, S) = a_{-}  + \sum_{t \in S} s_t. 
$$

These consitute at most $2^{|S_i|}$ different values, but some of them may be equal. Nonetheless, the values $s_i$ are between 0 and $1/q$, so for any value $x \in   [a_{-}, a_{+}]$, we can find a suitable subset $S$ such that $W^* = W(K_i, S)$ approximates $x$ within $1/(2q)$. In fact, for any ordering of the $s_t$ values $s_{t_1}, s_{t_2},\ldots, s_{t_{|S_i|}}$, looking at the two cumulative sums $s_{t_1} + s_{t_2} + \ldots + s_{t_k}$ scoring nearest to $x$, we can find one closer than $1/(2q)$. This shows that the desired subset can be found in $O(|S_i|)$ steps; with more work we can possibly find a better packing. 

For this choice of the values $w_i$, the sum of $w_i$ equals $W^*$, and $|W^*- J| \le 1/(2q)$.  To summarize this second case: Even if we can't find a dual variable $\lambda$ such that the constraint is satisfied, we can find one such that the sum of $w_i$ equals $W^*$, and $|W^*- J| \le 1/(2q)$. 

\end{paragraph}

In fact we have found the exact optimal weights for a slightly different problem.  Define the scaled weights $\widetilde{w_i} = {Jw_i}/{W^*}$,
then $\widetilde{w_i}$ solves the optimization problem \ref{P_avg} with the parameter $q$ replaced by $q^* = W^*q/J$.  This is clear if we notice that the optimization problem can be parametrized by $q$, and everything we've said so far holds for each fixed $q$. Thus, for any value of $W^*$, we just need to find a value $q^*$ such that the constraint holds. This amounts to
$$
 \sum_{i=1}^{J} \Phi(c_i) =  W^*q = q^*J
$$
and we find the claimed equation for $q^*$. Further $|q^*-q| = |q (J - W^*)/J| \le 1/(2J)$, as desired. This finishes the proof that $q^*$ exists, we just need to show how to find it. 

\begin{algorithm}[!ht]
\caption{Method to compute weights}\label{weights_alg}
\begin{algorithmic}[1]
\State $\eta_i, \sigma_i^2 \gets$ prior means and variances $i=1,\ldots,J$
\State  $q \gets $ significance threshold 
\If{condition \eqref{dual_variable_existence} holds} 
\State{Solve $\sum_i \Phi(c_1(\eta_i,\gamma_i,\lambda)) = Jq$ for $\lambda \in [1,\infty)$ using Newton's method} 
\State \textbf{return} $w_i = \Phi(c_1(\eta_i,\gamma_i,\lambda))/q$, $q^* =q$
\Else 
\State{Define $d(\lambda;\eta,\gamma) = \lambda\Phi(-c_1(\eta,\gamma; \lambda)) - \Phi\left(-(c_1(\eta,\gamma; \lambda) - \eta)/\gamma\right)  $}
\State \textbf{for all} {i} Solve $d(k_i;\eta_i,\gamma_i) =0$ for $k_i \in [l(\eta_i,\gamma_i),1]$ using Brent
\State{Let the sorted unique values $k_i$ be $K_{(i)}$}
\State{Define $w_i^{-}(\lambda) =
\left\{
	\begin{array}{ll}
		\Phi(c_1(\eta_i,\gamma_i; \lambda))/q  & \mbox{if  }\lambda \ge  k_i \\

		1/q & \mbox{if  }\lambda <  k_i
	\end{array}
\right. $}
\State{Define $s_i(\lambda) = I(\lambda=k_i) (1/q-w_i^{-}(\lambda))$}
\State{Define $w_i^{+}(\lambda)  = w_i^{-}(\lambda) + s_i(\lambda)$ }
\State{Define $W^{+}(\lambda)  = \sum_i w_i^{+}(\lambda)$ and $W^{-}(\lambda)  = \sum_i w_i^{-}(\lambda)$}
\State Find the index $j$ for which $W^+(K_{(j)}) \ge Jq > W^+(K_{(j+1)})$ via binary search
\If {$W^-(K_{(j)})  \ge Jq > W^+(K_{(j+1)})$}
 \State{Solve $W^+(\lambda) = Jq$ for $\lambda \in [K_{(j)},K_{(j+1)})$ using Brent} 
\State \textbf{return} $w_i = w^+_i(\lambda)$, $q^* =q$
\Else
\State Find the indices $S = \{j_1,\ldots,j_n\}$ such that $k_{j_i} = K_{(i)}$
\State Find the largest index $T$ such that $r^- := W^-(K_{(j)}) + \sum_{i \le T} S_{j_i}(K_{(j)}) \le Jq$
\State Define $r^+ = r^- + S_{j_{T+1}}(K_{(j)})$
\State Let $W^*$ be the closer of $r^-,r^+$ to $Jq$ (break ties arbitrarily)
\State Define $w_i^* = 
\left\{
	\begin{array}{ll}
		w^+_i(K_{(j)})  & \mbox{if  } i \notin S  \\
		w^+_i(K_{(j)})  & \mbox{if  } i = j_k, \mbox{ for } k>T+1 \\
		w^-_i(K_{(j)})  & \mbox{if  } i = j_k, \mbox{ for } k \le T \\
		w^\pm_i(K_{(j)})  & \mbox{if  } i = j_{T+1}, \mbox{ depending on the choice of } r^\pm \\
	\end{array}
\right. $
\State \textbf{return} $\widetilde{w_i} = {Jw^*_i}/{W^*}$, $q^* = W^*q/J$.
    \EndIf
\EndIf
\end{algorithmic}
\end{algorithm}

\paragraph{Final algorithm} The final algorithm is sumarized in Algorithm \ref{weights_alg}. We compute the values $k_i$ for each $i$, using Brent's method on the equation $d(k)=0$. This takes $O(J)$ steps.  Next, we sort the unique values $k_i$ ($O(J \log J)$ steps), and do a binary search on the values $W(k_i)$ to find the right interval for $\lambda$. Calculating a value $W(\lambda)$ takes $O(J)$ steps. The binary search takes $O(\log J)$ function evaluations of $W$, making the total cost of this step $O(J\log J)$. 

Once we found the right interval, we have to deal with the two possibilities identified above. If there is an exact solution $W(\lambda)=J$, then after finding the right interval, Brent's method is used to solve the equation ($O(J )$ steps). On this interval, the function $W$ is smooth with explicitly computable derivatives. If there is no exact solution, then we simply return the closest interval endpoint ($O(J )$ steps). The overall cost is $O(J \log J)$. 

\section{Data sources}

\label{detailed_data}
\subsection{CARDIoGRAM - CAD}
This is a meta-analysis of 14 coronary artery disease (CAD) GWAS, comprising 22,233 cases and 64,762 controls of European descent \citep{schunkert2011large}. The study includes 2.3 million SNPs. In each of the 14 studies and for each SNP, a logistic regression of CAD status was performed on the number of copies of one allele, along with suitable controlling covariates.
The resulting effect sizes were combined across studies using fixed effects or random effects meta-analysis with inverse variance weighting. 

\subsection{C4D - CAD}

This is a meta-analysis of 5 heart disease GWAS, totalling 15,420 CAD cases and  and 15,062 controls \citep{coronary2011genome}. The samples did not overlap those from CARDIoGRAM. The analysis steps were similar to CARDIoGRAM. 

\subsection{Chronic Kidney Disease Consortium (CKDGen) - eGFR creatinine}

This is a GWAS of kidney traits in
67,093 participants of European ancestry from 20 population-based cohorts \citep{kottgen2010new}. eGFR creatinine (eGFRcrea) was the trait with the largest sample size. There is no reported overlap with the samples from C4D. The analysis steps were similar to the previous two studies.

\subsection{Blood Lipids}
This is a GWAS of blood lipids in a sample from European populations \citep{teslovich2010biological}. Triglyceride levels (TG) were one of the traits, with sample size 96,598, chosen here out of all lipids because of its previous appearance in \citep{andreassen2013improved}. Standard protocols for GWAS were used: linear regression analysis with study-specific covariates, combined using fixed-effects meta-analysis. 
 
\subsection{Psychiatric Genomics Consortium - Schizophrenia}

This is a mega-analysis (i.e. using the raw data not just summaries) combining GWAS data from 17 separate studies of schizophrenia (SCZ), with a total of 9,394 cases and 12,462 controls \citep{schizophrenia2011genome}. They tested for association using logistic regression of SCZ status on the allelic dosages. 
 The overlap with the blood lipids study consists of 1,459 controls ($12\%$ of controls in the SCZ study), from the British 1958 Birth Cohort of the Wellcome Trust Case Control Consortium. 

\bibliographystyle{apalike}
\bibliography{weighted}

\end{document}